\documentclass[a4paper,11pt]{article}

\usepackage{enumerate}
\usepackage{latexsym}
\usepackage{amsmath}
\usepackage{amssymb}
\usepackage{soul}
\usepackage{amsthm}
\usepackage{epsfig}
\usepackage[latin1]{inputenc}
\usepackage{color}
\usepackage{authblk}

\newtheorem{theorem}{Theorem}[section]
\newtheorem{lemma}[theorem]{Lemma}
\newtheorem{corollary}[theorem]{Corollary}

    \def\independenT#1#2{\mathrel{\setbox0\hbox{$#1#2$}%
    \copy0\kern-\wd0\mkern4mu\box0}}

\def\t{\tau}

\def\l{\ell}

\def\cal{\mathcal}

\renewcommand{\epsilon}{\varepsilon}

  \title{Coalition Games on Interaction Graphs: A Horticultural Perspective}
  \author[1,4]{Nicolas Bousquet}
  \author[3]{Zhentao Li}
  \author[1,2]{Adrian Vetta}
  \affil[1]{Department of Mathematics and Statistics, McGill University}
  \affil[2]{ School of Computer Science, McGill University }
  \affil[3]{ \'Ecole Normale Sup\'erieure de Paris. }
  \affil[4]{ Group for Research in Decision Analysis (GERAD), HEC Montr\'eal}

\begin{document}

\maketitle

\begin{abstract}
We examine cooperative games where the viability of a coalition
is determined by whether or not its members have the ability to
communicate amongst themselves independently of non-members.
This necessary condition for viability was proposed by Myerson in \cite{Mye77} and
is modeled via an interaction graph $G=(V,E)$; a coalition $S\subseteq V$
is then viable if and only if the induced graph $G[S]$ is connected.
The non-emptiness of the core of a coalition game can be tested by a well-known covering LP.
Moreover, the integrality gap of its dual packing LP defines exactly the
multiplicative least-core and the relative cost of stability of the coalition game.
This gap is upper bounded by the packing-covering ratio which,
for graphical coalition games, is known to be at most 
the treewidth of the interaction graph plus one \cite{MER13}. 

We examine the packing-covering ratio and integrality gaps of graphical coalition
games in more detail. We introduce the thicket parameter of a graph, and 
prove it precisely measures the packing-covering ratio. It also
approximately measures the primal and dual integrality gaps.
The thicket number provides an upper bound
of both integrality gaps. Moreover we show that for any interaction graph,
the primal integrality gap is, in the worst case, linear in terms 
of the thicket number while the dual integrality gap is polynomial in terms of it.
At the heart of our results, is a graph theoretic minmax theorem showing the
thicket number is equal to the minimum width of a vine decomposition of the
coalition graph (a vine decomposition is a generalization of a tree decomposition). 
We also explain how the thicket number relates to the 
VC-dimension of the set system produced by the game.
\end{abstract}

%vine is a broadleaved deciduous hardwood tree of the genus Betula

%"The main fundamental question in cooperative game theory is the question how to allocate the total 
%generated wealth by the collective of all players?the player set N itself?over the different players in 
%the game. In other words, what binding contract between the players in N has to be written? 
%Various criteria have been developed."
%
%
%"There is an alternative approach to the definition of the Core of a cooperative game. 
%This approach is based on the domination of one imputation by another. The notion of 
%domination has its roots in the seminal work on bargaining processes by Edgeworth (1881). 
%Edgeworth considered coalitions of traders exchanging quantities of commodities in an 
%economic bargaining process that is not based on explicit price formation, also known as 
%Edgeworthian barter processes. Later this was interpreted as the ?Core of an economy? 
%(Hildenbrand and Kirman, 1988)."
%
%Upon the development of game theory, and cooperative game theory in particular, 
%the fundamental idea of blocking was formulated by Gillies (1953) in his Ph.D. dissertation. 
%This idea was developed further in Gillies (1959) and linked with the work of Edgeworth (1881) 
%by Shubik (1959). Since Shubik?s unification, the theory of the Core?both in cooperative game theory 
%as well as economic general equilibrium theory?took great flight. - Gilles quotes.

% epsilon-core/least-core is additive

\section{Introduction}

At the heart of cooperative game theory is the problem of how a group of agents
should share the wealth that they collectively create. Its foremost concept 
is the {\em core} whose roots date back to Edgeworthian bargaining and cooperative improvement 
(\cite{Edge1881}; see also \cite{HK88}).
It was first formalized by Gillies in~\cite{Gillies53,Gillies59} via a {\em coalition game} $\cal{G}=(I,v)$ with a
 set $I$ of agents and a valuation function $v:2^I \rightarrow \mathbb{N}$.
The {\em core} of the coalition game is the set of feasible solutions~to:

\[
\arraycolsep=3.2pt\def\arraystretch{1.4}
\begin{array}{rcl@{\extracolsep{20pt}}l}
 \sum\limits_{i: i \in I} x_i & = &  v(I) &\\
 \sum\limits_{i: i \in S} x_i &\ge& v(S) & \forall S \subset I \\
x_i                        &\ge& 0
\end{array}
\]

Informally, we are allocating $x_i$ to agent $i$ and $v(S)$ represents the amount of wealth that the coalition $S$ can generate by itself. 
Consequently, the coalition $S$ will {\em block} any distribution scheme that does not allocate
its members at least $v(S)$ in total. Thus the wealth $v(I)$ of the  \emph{grand coalition} 
must be distributed in such a fashion that no coalition wishes to block the allocation.
The core is the set of vectors of payoffs that have this property.

This definition immediately prompts two questions: (i) What processes enable the formation of coalitions?
(ii) Even if coalitions can form and negotiate, do core solutions exist?
Concerning the former question, it is unrealistic to assume that every subset of agents has the ability to act as
a collective. Indeed, Myerson~\cite{Mye80} argued that feasible coalitions require structural 
properties that enable them to function. Clearly, one necessary property is that ``communication" is
possible between members of the coalition and \cite{Mye77}  formalized this ability 
using an {\em interaction (communication) graph} $G=(I,E)$. Here a pair of agents induces
an edge in $G$ if they are able to interact and a coalition $S$ is feasible if $S$ induces
a connected subgraph of $G$. Observe that two members of a feasible coalition do not need to be able
to interact directly but they must be able to communicate indirectly via chains consisting of other 
members of the coalition. Thus, a coalition $S$ is \emph{viable}
 if and only if the induced subgraph $G[S]$ is connected -- in particular $v(S)=0$
 when $G[S]$ is disconnected. Such {\em graphical coalition games} are the focus of this 
 paper.\footnote{See \cite{Gilles10} for a general introduction to cooperative games on networks.}
% (We remark that if all coalition members can communicate {\em directly} then
%$G$ corresponds to a partition into disjoint cliques; thus the classical work of \citet{AD74} 
%can also be modelled by an interaction graph. {\bf [Check this...]} {\bf [I don't think this is true. 
%In general, it can be any union of cliques. E.g., chordal graphs or triangle-free graphs with only coalitions of size 2.]})

For the latter question, the core is often empty. 
Indeed, it is straightforward to verify that the core of the game $\cal{G}$ is non-empty {\em if and only if} the following 
primal linear program 
has an optimal (fractional) 
solution whose value $\kappa^f(\mathcal{G})$ equals $v(I)$. \vspace{-7pt}

\[
\arraycolsep=3.2pt\def\arraystretch{1.4}
\begin{array}{rrrcl@{\extracolsep{10pt}}l}
\textrm{{\tt \small Covering-LP:}}
& \min        & \sum\limits_{i \in I} x_i   &     &\\
&  \text{s.t.} & \sum\limits_{i: i \in S} x_i & \ge & v(S) & \forall S \subseteq I \\
& &                        x_i & \ge & 0
\end{array}
\]
Interestingly, there is an elegant graphical characterization for when a graphical coalition 
game has a non-empty core for all possible valuation functions $v$. Namely,
a graphical coalition game is {\em strongly balanced} if and only if the interaction graph $G$ is a 
forest  \cite{LOW92}.

\subsection{The Least-Core and the Relative Cost of Stability}
Given the possible emptiness of the core, it is natural to consider solutions where the
core constraints in this {\tt Covering-LP} are relaxed. Specifically, for each feasible 
coalition $S$, given $\alpha\ge 1$, there is 
a constraint  
$\alpha \cdot \sum_{i: i \in S} x_i \ge  v(S)$.
These constraints imply that a coalition $S$ will not block an allocation unless it
can unilaterally improve its total wealth by more than an $\alpha$ factor.
The set of feasible solutions then form the $\alpha${\em-core}. The minimum $\alpha$ for which the 
$\alpha$-core is non-empty arises when $\alpha^*=\frac{\kappa^f(\cal{G})}{v(I)}$. 
The $\alpha^*$-core is called the {\em (multiplicative) least-core}.\footnote{The least-core is often defined 
with respect to an
additive ($\epsilon$), not a multiplicative ($\alpha$), guarantee. Since additive guarantees are not scale 
invariant, it is preferable here to focus upon multiplicative guarantees.}
Moreover, we can also determine the least-core by considering the dual of the primal linear program.

\[
\arraycolsep=3.2pt\def\arraystretch{1.4}
\begin{array}{rrrcl@{\extracolsep{10pt}}l}
\vspace{-7pt} \textrm{{\tt \small Packing-LP:}}
& \max        & \sum\limits_{S: S\subseteq I} v(S)\cdot y_S && \\
& \text{s.t.} & \sum\limits_{S\subseteq I: i\in S} y_S       &\le& 1 &\forall i \in I \\
&             & y_S                                      &\ge& 0 &\forall S\subseteq I
\end{array}
\]

Let $\rho^f(\mathcal{G})$ be the optimal fractional solution to this {\tt Packing-LP}, and
let $\rho(\mathcal{G})$ be the optimal integral solution. Then $\alpha^*$ is exactly equal to the {\em dual integrality
gap} $\frac{\rho^f(\mathcal{G})}{\rho(\mathcal{G})}$. To verify this, observe that
$\rho^f(\mathcal{G})=\kappa^f(\mathcal{G}) =\alpha^*\cdot v(I)$, by strong duality.
But the optimal integral solution to the dual has value $\rho(\mathcal{G})=v(I)$; simply 
set $y_I=1$ and $y_S=0$ for every $S \neq I$.
Here we are making the standard assumption in the literature that, for any coalition game, the valuation
function $v$ is superadditive.
%\footnote{If $v$ is not superadditive then the {\em superadditive cover} $\hat{v}$ (see \citet{AD74}) is typically taken as the valuation function. The superadditive cover is defined to satisfy the partition property:  $\hat{v}(S)= \max_{\mathcal{Z}} \sum_{Z_i\in \mathcal{Z}} v(Z_i)$ where the maximum is taken over all partitions $\mathcal{Z}=\{Z_1, Z_2,\dots , Z_k\}$ of the agents~$S$. We remark that the concept of balancedness is equivalent to the superadditivity of the corresponding replica game \cite{Wood83}.} 
In particular, property reflects the simple observation that any collective has the option to voluntarily 
partition itself into subgroups to generate wealth.
Thus, $\alpha^*=\frac{\rho^f(\mathcal{G})}{\rho(\mathcal{G})}$.

%To understand this, we may combinatorially interpret the optimal integral solution $\rho(\mathcal{G})$ to the 
%dual as follows. Let $\mathcal{S}=\{S_1, S_2,\dots , S_k\}$ denote a partition of the agents $I$. We say that the 
%value of a partition $\mathcal{S}$ of $I$ is
%\begin{equation*}
%\rho(\mathcal{S}) =\sum_{S_i\in \mathcal{S}} v(S_i)
%\end{equation*}
%and the {\em partition value} of the whole game $\mathcal{G}$ is the maximum value of any partition.
%Thus,
%$$\rho(\mathcal{G}) = \max_{\mathcal{S}} \rho(\mathcal{S})$$
%{\bf [Clean up previous statements about wealth of $I$ to be consistent with $\rho$...]}

This paper will study least-cores in coalition games over interaction graphs.
Interestingly, multiplicative least-cores are equivalent to the concept of the {\em relative cost of stability}.
Bachrach et al.~\cite{BEM09} asked how much it would cost (an external authority) 
to stabilize a coalition game; $i.e.$ what is the minimum total payment required such that no coalition 
can benefit by blocking the allocation. The relative cost of stability~\cite{MRM11} is then
defined to be the ratio between this minimum payment and the total wealth the grand coalition
can generate, namely $\frac{\kappa^f(\mathcal{G})}{v(I)}=\alpha^*$. Thus, the 
relative cost of stability is also given by the dual integrality gap.

\subsection{Our Results}
So to develop an understanding of coalition games, we must study the 
primal and dual linear programs. In particular, we will focus on
the graphical coalition games of \cite{Mye77}.
Specifically, we are interested in how the primal and dual integrality
gaps vary with the topology of the interaction graph.

As inferred by our nomenclature, the primal and dual form a pair of packing and covering linear 
programs. Thus a natural starting point is to consider the {\em packing-covering ratio}
of a game $\cal{G}$; this is the ratio $\frac{\kappa(\cal{G})}{\rho(\cal{G})}$ 
between the values of the optimal integral solutions to the primal and the dual.
Observe that, by strong duality, the packing-covering ratio is the product of the primal integrality gap
$\frac{\kappa(\cal{G})}{\kappa^f(\cal{G})}$ 
and the dual integrality gap $\frac{\rho^f(\cal{G})}{\rho(\cal{G})}$. Consequently, the packing-covering
ratio trivially upper bounds both integrality gaps. Packing-covering ratios have been studied extensively in
graph theory. Special attention has focused on problems with the Erd\H{o}s-P\'osa property, 
where the ratio is a function of the packing number and is otherwise independent of 
the graph. Interestingly, whilst graphical coalition games do not have the Erd\H{o}s-P\'osa property,
the packing-covering ratio can be bounded by an important parameter of the interaction graph, namely, treewidth.
Indeed, Meir et al.~\cite{MER13} proved that, for any valuation function (game $\cal{G}$) 
over an interaction graph $G$, the packing-covering ratio is at most the treewidth $\omega(G)$ plus one. 
%Furthermore, they show there exists a graph and a valuation function for which this bound is tight.

We extend the work of Meir et al. in several ways. 
First, we show that structurally treewidth is not the most appropriate invariant in understanding the 
packing-covering ratio. The topological parameter that corresponds exactly to the packing-covering ratio
is a concept we term the {\em thicket number} of the graph. 
Specifically, in Section \ref{sec:ratio}, we show that for every coalition game $\cal{G}$ over a graph $G$ 
the packing-covering ratio is at most the thicket number, $\tau(G)$, of the graph. 
Conversely, for {\em every} graph $G$ there exists a coalition game $\cal{G}$ 
for which the packing-covering ratio is at least the thicket number.

\begin{theorem}\label{thm:1}
For any interaction graph $G$, the packing-covering ratio satisfies:
$$\tau(G)  \ \ \le_{\exists}\ \  \frac{\kappa(\cal{G})}{\rho(\cal{G})} \ \ \le_{\forall}\ \  \tau(G)  $$
\end{theorem}
Observe that, in order to concisely formulate our results, we use the notation $\leq_\exists$ and $\leq_\forall$.
Here $\leq_\exists$ means that
{\em there exists} a game $\mathcal{G}$ over the interaction graph $G$ such that the inequality is satisfied, and
$\leq_\forall$ means that {\em for every} game $\mathcal{G}$ 
over $G$ the inequality is satisfied.\footnote{We remark that inequalities of the form $\geq_\exists$ and $\geq_\forall$ are 
not interesting from a game-theoretic perspective.
Indeed, for any graph $G$, we have that $\frac{\kappa(\cal{G})}{\rho(\cal{G})} \ge_{\forall} 1$, by weak duality.
Furthermore, the packing-covering ratio of any coalition game with one viable coalition trivially equals $1$. 
So $1 \geq_{\exists} \frac{\kappa(\cal{G})}{\rho(\cal{G})}$.}

Theorem \ref{thm:1} relies on a graphical minmax result that we prove in Section \ref{sec:thicket}.
Specifically, we show that thickets have a dual notion called {\em vine decompositions}.
These decompositions can be viewed as a ``thin" relative of tree decompositions.
In particular, the vinewidth of a graph is at most the treewidth plus one, and is typically smaller.

%To begin, we focus on {\em simple games} where the value of each coalition is in $\{0,1\}$. 
%Thus a solution of the \textsc{Covering LP} 
%problem is a subset of vertices $X$ such that each coalition of value $1$ contains a vertex of $X$.

In principle, the primal and dual integrality gaps could be much less than the thicket
number. However, we prove the thicket number is (approximately) the correct
measure for these integrality gaps as well.
Specifically, in Section \ref{sec:primal-gap} we prove
\begin{theorem}\label{thm:2}
For any interaction graph $G$, the primal integrality gap satisfies:
$$\frac14 \tau(G)  \ \ \le_{\exists}\ \  \frac{\kappa(\cal{G})}{\kappa^f(\cal{G})} \ \ \le_{\forall}\ \  \tau(G)  $$
\end{theorem}
Interestingly, unlike for the packing-covering ratio, the upper and lower bounds cannot be closed
completely for the primal integrality gap. Indeed, for any graph $G$ there is a constant 
$a_G$ 
such that $a_G\cdot \tau(G) \le_{\exists} \frac{\kappa(\cal{G})}{\kappa^f(\cal{G})} \le_{\forall} a_G\cdot \tau(G)$. 
However $a_G$ really does vary with the graph. In particular, we show that 
$a_G\rightarrow 1$ for the family of graphs that correspond to the powers of paths. 
On the other hand, we prove that $a_G\le \frac12$ for cliques.
 It follows that the constant $1$ in the upper bound in Theorem \ref{thm:3} cannot be 
decreased, whilst the constant in the lower bound cannot be increased above $\frac12$.

%We also study the relevance of the upper and lower bounds in both cases. For the primal integrality gap
%we show that there exist a family $\cal{F}$ of graphs $G$ (namely the power of paths) and of games $\cal{G}$ 
%over interaction graphs in $\cal{F}$ such that $\frac{\kappa(\cal{G})}{\kappa^f(\cal{G})} \longrightarrow \tau(G)$.
%It implies that the upper bound cannot be improved. Moreover, for cliques {\bf  and grids ?}, we show that
%the primal integrality gap is bounded above by $\tau(G)/2$ for any game, which ensures that the lower bound
%cannot be improved above $\tau(G)/2$.

Next consider the dual integrality gap. In Section \ref{sec:dual-gap} we prove
\begin{theorem}\label{thm:3}
There exist $c$ and $\delta$ such that for any interaction graph $G$, the dual integrality gap satisfies:
$$ c \cdot \tau(G)^\delta  \ \ \le_{\exists}\ \  \frac{\rho^f(\cal{G})}{\rho(\cal{G})} \ \ \le_{\forall}\ \  \tau(G)  $$
\end{theorem}

%Note that the lower bound of Theorem~\ref{thm:3} on the primal integrality gap 
%is much better than the lower bound of Theorem~\ref{thm:2} on the dual integrality gap. It can be explained 
%by the fact that, since thickets are defined as hitting sets of connected subgraphs, they are easier
%to manipulate to find coverings than packings. Also remark that all the lower bounds are obtained for simple games
%($i.e.$ games where coalitions have value $0$ or $1$).
Again, it is not possible to close the upper and lower bounds for the dual integrality gap completely. Even more interestingly, the polynomial range between the upper and lower bounds is necessary. This is since for the family of grid graphs, the exponent of $\tau(G)$ in the lower bound must be $1$, whereas for the family of cliques the exponent of $\tau(G)$ in the upper bound is $\frac12$. It follows that the exponent $1$ in the upper bound of Theorem \ref{thm:2} cannot be decreased, whilst the exponent $\delta$ in the lower bound cannot be increased above $\frac12$.
%Again, it is not possible to close the upper and lower bounds for the dual integrality gap completely. Even more interestingly, the polynomial range between the upper and lower bounds is necessary.
%To see this, for any graph $G$ there are constants $c_G$ and $\delta_G$  
%such that $c_G\cdot \tau(G)^{\delta_G} \le_{\exists} \frac{\kappa(\cal{G})}{\kappa^f(\cal{G})} \le_{\forall} c_G\cdot \tau(G)^{\delta_G}$. 
%We show that for the family of grid graphs the exponent $\delta_G=1$. 
%For cliques, however the exponent $\delta_G=\frac12$, that is, the 
%upper bound is $\cal{O}(\sqrt{\tau(G)})$.
% It follows that the exponent $1$ in the upper bound of Theorem \ref{thm:2} cannot be 
%decreased, whilst the exponent $\delta$ in the lower bound cannot be increased above $\frac12$.
We remark that the value of $\delta$ in Theorem~\ref{thm:2} relies on the existence of a grid minor of
polynomial size in the treewidth of the graph, a deep result of \cite{CC13}. Determining the best possible
order of the polynomial in the lower bound of Theorem~\ref{thm:2} (between $\delta$ and $\frac{1}{2}$) is 
an interesting open question.

Finally, in Section~\ref{sec:VC}, we show how the VC-dimension may also be used to bound
the packing-covering ratio and the integrality gaps. Given our previous discussion,
the resultant bounds must be weaker than those obtainable via the thicket number.
Indeed, we show how these VC-dimension bounds may can be derived from the 
thicket number bounds.

%These results are summarized in Table \ref{tab:results}.
%\begin{center}
%\begin{tabular}{|l|l|}\hline
%  Setting & Bound\\\hline
%  $\forall G \forall v$ & $\RCoS(G,v) \le \omega(G)+1$ \cite{MER13} \\
%  $\forall k \exists G \exists v$ & $\RCoS(G,v)\ge k+1 = \omega(G)+1$ \cite{MER13} \\\hline
%  $\forall G \forall v$ & $\RCoS(G,v) \le \ttw(G)$ \\
%  $\forall G \exists v$ & $\RCoS(G,v) \ge \ttw(G)$ \\
%  $\forall G \forall v$ & $\RCoS_f(G,v) \le \omega(G)+1$ \\
%  $\forall G \exists v$ & $\RCoS_f(G,v) \ge (\omega(G)+1)/4 \ge (\ttw(G)+1)/4$ \\
%  $\exists G \forall v$ & $\RCoS_f(G,v) \le (\omega(G)+1)/4 + O(1)$ \\
%  $\forall k \forall \epsilon \exists G \exists v$ & $\RCoS_f(G,v) \ge (1-\epsilon)\ttw(G) = (1-\epsilon)\omega(G) = (1-\epsilon)k$ \\\hline
%\end{tabular}\label{tab:results}
%\end{center}

%[\textbf{Open question:} $\exists G \forall v \RCoS_f(G,v) \le (\ttw(G)+1)/2$? Or can be instead 
%improve $\RCoS_f(G,v) \ge (\ttw(G)+1)/4$, to $\RCoS_f(G,v) \ge (\ttw(G)+1)/2$ for example?]

\section{Thickets and Vines}\label{sec:thicket}
Recall that \cite{MER13} show that the packing-covering ratio can
be bounded in terms of the treewidth of the interaction graph $G$. 
To understand this, we begin with a brief review of treewidth.
We will then introduce a better fitting parameter for analyzing coalition games.

\subsection{Tree Decompositions and Brambles}
Treewidth provides a measure of how closely a graph shares some structural
separation properties possessed by trees. 
Formally, given an undirected graph $G=(V,E)$ we may represent it by a tree $T=(N,L)$
and a labeling $\l:N\rightarrow 2^V$.\footnote{For clarity, 
we will refer to vertices and edges in $G$ and
nodes and links in $T$.}
The labeling assigns to each node $t\in T$ a subset $\l(t)=V_t$ of vertices of $G$.
For each $v\in V$ we denote by $T_v$ the set of nodes in $T$ for which $v$ is
included in the label, $i.e.$ $T_v=\{t: v\in V_t\}$.
 We say that a tree and labeling, $(T,\l)$, is a {\em tree decomposition}
of $G$ if: \\
 (i) For each vertex $v$ of $G$, the set $T_v$ is a {\bf non-empty} and {\bf connected} subgraph of $T$. \\
 (ii) For each edge $e=(u,v)$ in $G$, the subtrees $T_u$ and $T_v$ {\bf intersect} in $T$. 
 
 %$i.e.\ T_u \cap T_v \neq \emptyset$.
The {\em width} of a tree decomposition $(T,\l)$ of $G$ is the size of the largest label of
a node in $T$ {\em minus one}.\footnote{The decision to subtract one was made to ensure trees have treewidth one. 
Unfortunately, as is apparent, this choice leads to an unaesthetic ``plus one" in many theorems concerning treewidth.} 
The {\em treewidth}, $\omega(G)$, is the minimum width of a
tree decomposition of $G$.
Meir et al. \cite{MER13} show that treewidth relates to coalition games via the following bound.
\begin{theorem}{Meir et al. \cite{MER13}}\label{thm:tw}
For any interaction graph $G$, the packing-covering ratio satisfies:
$$\frac{\kappa(\cal{G})}{\rho(\cal{G})} \ \ \le_{\forall}\ \  \omega(G)+1  $$
\end{theorem}
To delve further into this topic, it is important to note that there 
are combinatorial structures called brambles that provide a dual notion for tree decompositions.
A \emph{bramble} is a collection $\mathcal{F}=\{F_1,F_2,\dots,F_p\}$ of sets such that: \\
(i) Each $F_i\subseteq V$  induces a connected subgraph of $G$, and \\
(ii) Every pair $F_i$ and $F_j$ in $\mathcal{F}$ {\em either} intersect (share a vertex) {\em or}  
are adjacent (there is an edge with one endpoint in $F_i$ and one endpoint in $F_j$). 

The {\em hitting size} $\beta(\cal{F})$ of a bramble $\cal{F}$ is
the minimum size of a subset of vertices that intersects each set $F_i$ in $\cal{F}$.
The {\em bramble number} $\beta(G)=\max_{\mathcal{F}} \beta(\cal{F})$ 
is the maximum hitting size of any bramble in $G$.
Seymour and Thomas~\cite{ST93} proved the following minmax theorem:
%{\bf[Explain as min-max?...]}
\begin{theorem}\label{thm:seymour}\cite{ST93}
 The bramble number $\beta(G)$ is equal to the treewidth $\omega(G) - 1$.
\end{theorem}

Brambles (rather than tree decompositions) directly relate to
coalition games. Moreover, the relationship is actually through combinatorial
structures we call thickets.

\subsection{Thickets}
Let $G=(V,E)$ be an undirected graph.
A {\em thicket} $\cal{H}=\{H_1,H_2,\dots,H_p\}$ is a collection 
of  sets such that: \\
(i) Each $H_i\subseteq V$ induces a connected subgraph of $G$, and \\
(ii) Every pair $H_i$ and $H_j$ in $\mathcal{H}$ intersects. 

Observe that thickets differ from brambles in that they must pairwise intersect -- adjacency 
is not sufficient.
The {\em hitting size} $\tau(\cal{H})$ of a thicket $\cal{H}$ is
the minimum size of a vertex set that intersects each set $H_i$ in $\cal{H}$.
The {\em thicket number} $\t(G)=\max_{\mathcal{H}} \tau(\cal{H})$ 
is the maximum hitting size of any thicket in $G$. 

Intuitively, thickets are indeed the objects that directly correspond to the packing-covering ratio.
The packing number of a thicket $\mathcal{H}$ is exactly one, but its covering number 
is $\tau(\mathcal{H})$. We will formalize this intuition in Section \ref{sec:ratio}.
First, let's see three simple classes of graphs that will illustrate the concept of thickets, and which 
will also be very useful in the technical results that follow.\\

\noindent{\tt Example 1: Trees.} Let $G=T_n$ be a tree on $n$ vertices. 
A thicket $\mathcal{H}=\{H_1,\dots, H_p\}$ on $T_n$ then consists of a collection of pairwise intersecting
subtrees. It is well-known, by the {\em Helly Property} of trees, that such a collection must contain
a common vertex. Thus the thicket number, $\t(T_n)$, of a tree $T_n$ is at most $1$.\\

\noindent{\tt Example 2: Cliques.} Let $G=K_n$ be a clique on $n$ vertices.  
A thicket $\mathcal{H}=\{H_1,\dots, H_p\}$ on $K_n$ then consists of a collection of pairwise 
intersecting subcliques. Suppose the smallest of these cliques, say $H_1$, has cardinality at most 
$\lceil \frac12 n\rceil$. Then, as $H_1$ itself intersects all of the sets in $\mathcal{H}$, we have a hitting set of
cardinality $\lceil \frac12 n\rceil$. Otherwise all the cliques have cardinality at least $\lceil \frac12 n\rceil+1$.
But then any set $X$ of cardinality $\lceil \frac12 n\rceil$ is a hitting set for $\mathcal{H}$. Thus
$\tau(G)\le \lceil \frac12 n\rceil$.\\

\noindent{\tt Example 3: Grids.}
Let $G=R_k$ be a $k\times k$ grid graph -- the planar graph formed by a grid of $k$ rows and $k$ columns.
Consider the thicket $\cal{H}$ defined as follows. We have a set $H_{R,C}=R\cup C$, for each row $R$ and each column $C$ in the grid.
Clearly each set is connected. Moreover, each pair of sets intersect. So $\cal{H}$ is a thicket. Now take any vertex set $X$ of cardinality less than $k-1$.
Since there are $k$ rows, $X$ must miss some row $\hat{R}$; similarly it must miss some
column $\hat{C}$. Hence, $X$ is not a hitting set for $\cal{H}$ as it does not intersect $H_{\hat{R},\hat{C}}$.
So the hitting size $\t(\mathcal{H})$ is at least $k$. Consequently,
the thicket number, $\t(R_k)$, of the grid $R_k$ is at least $k$.

\subsection{Vine Decompositions}
Brambles are dual to tree decompositions, and thickets also have
dual structures. Since the definition of a thicket is more stringent than
that of a bramble, it must be the case that the definition of its dual structures
is more relaxed than that of a tree decomposition. In particular,
its dual will be a {\bf thin} tree (let's call it a {\em vine}!). 

Formally, given an undirected graph $G=(V,E)$, we construct a 
(vine) tree $T=(N,L)$.
The labeling assigns to each node $t\in T$ a subset $\l(t)=V_t$ of vertices of $G$.
For each $v\in V$ we denote by $T_v$ the set of nodes in $T$ for which $v$ is
included in the label, $i.e.$ $T_v=\{t: v\in V_t\}$.
We say that a tree and labeling, $(T,\l)$, is a {\em vine decomposition}
of $G$ if:
\begin{enumerate}[(i)]
\item For each vertex $v$ of $G$, the set $T_v$ is a {\bf non-empty} and {\bf connected} subgraph of $T$.
\item For each edge $e=(u,v)$ in $G$, the subtrees $T_u$ and $T_v$ {\bf intersect or are adjacent} 
in $T$.
\end{enumerate}

The {\em width} of a vine decomposition $(T,\l)$ of $G$ is the size of the largest label of
a node in $T$. The {\em vinewidth}, $\nu(G)$, is the minimum width of a
vine decomposition of $G$. 
The main structural result of the paper is that the thicket number $\tau(G)$ is 
equal to the vinewidth $\nu(G)$. Before proving this result in Section \ref{sec:thicket-duality}, we
will develop some understanding of vine decompositions.
First, let's return to the simple examples of trees, cliques and grids.\\

\noindent{\tt Example 1: Trees.} Observe that a tree $G=T_n$ gives a trivial
vine decomposition of itself, that is $T=(T_n, \ell)$ where $\ell(v)=\{v\}$. 
This is a vine decomposition as for each edge $(u,v)$ in $G$ the two vertices
are clearly still adjacent in $T$. Thus the vinewidth of a tree is at most $1$. \\

\noindent{\tt Example 2: Cliques.} Let $G=K_n$ be a clique on $n$ vertices. 
The clique has vinewidth at most$\lceil \frac12 n\rceil$. The 
corresponding vine decomposition has two nodes, each containing (roughly) half the vertices. 
This is a vine decomposition as for each edge $(u,v)$ in $G$ the two vertices
are either in the same node of $T$ or in adjacent nodes. \\

\noindent{\tt Example 3: Grids.}
Let $G$ be a $k\times k$ grid graph. Let the (vine) tree $T$ be a path on $k$ nodes. 
Let the $i$th node in the path satisfy $\l(i)=C_i$, where $C_i$ is the set of vertices in the
$i$th column of $G$. This is a vine decomposition. For each $v\in V$ the set $T_v$ 
is a singleton node in $T$, and is thus non-empty and connected. 
For each edge $e=(u,v)$ in $G$, either $T_u=T_v$ if $u$ and $v$ are in the same column 
of $G$, or $T_u$ and $T_v$ are adjacent nodes in $T$ if $u$ and $v$ are in the same row of $G$.
Clearly the width of this vine decomposition is $k$ and, thus, the vinewidth $\nu(G)$ is 
at most $k$. \\

Now treewidth and vinewidth are at most a multiplicative factor two apart.
\begin{theorem}\label{thm:approx}
Vinewidth and treewidth are related by
$\nu(G)-1 \le \omega(G) \le 2\nu(G)-1$. Moreover there exist graphs for which the lower and 
upper bounds are tight.
\end{theorem}
\begin{proof}
Observe that a tree decomposition of $G$ is a vine decomposition of $G$.
Thus, $\nu(G) \le \omega(G)+1$. 
Note that this bound is almost tight for the grid $R_k$; we have $\nu(R_k)=k$ 
(this follows as the thicket number equals the vinewidth (see Theorem 2.5 below)) 
and it is well-known that $\omega(R_k)=k$. In fact, 
graphs can be constructed for which this lower bound is exactly tight; 
we omit the details.

On the other hand, given a vine decomposition $(T, \ell)$
we can create a tree-decomposition $(\hat{T}, \hat{\ell})$ by augmenting it as follows. 
We replace each link in $T$ by a path of length two in $\hat{T}$.
For each new node $t\in \hat{T}$ in the middle of the path that replaced the link $(t_1, t_2)\in T$ we
set $\hat{\ell}(t)= \ell(t_1)\cup \ell(t_2)$. It is easy to verify that this is a 
tree-decomposition.
Furthermore the width of this tree-decomposition is at most $2\nu(G)-1$ 
(since each label has size at most $2 \nu(G)$). This upper bound is tight for cliques
as $K_n$ has treewidth $n-1$ and vinewidth $\lceil \frac{n}{2} \rceil$.
 \end{proof}

As with tree decompositions, an important property of vine decompositions is that
nodes in the (vine) tree correspond to separators in the original graph.

%Take a node $t\in T$ with label $V_t$. Let $T_1,T_2,\dots,T_r$ be 
%the subtrees of $T$ formed by the removal of $t$.
%Let $A_i= \bigcup_{x\in T_i} W_x - V_t$, for all $1\le i\le r$,
%and let $\bar{A}_i = A_1\cup \cdots \cup A_{i-1}\cup A_{i+1}\cup \cdots \cup A_r$.
%Then $A^*_i=A_i-\bar{A}_i$ consists of vertices of $G$ that appear only in labels
%of nodes in $T_i$. Let $Z$ be the set of vertices that appear in at least two
%of the $A_i$, that is $Z= \bigcup_{1\le i< j \le r} (A_i \cap A_j)$.
%
%Now, for any vertex $z\in Z$ we have that $\G(z)\subseteq V_t$.
%For example, suppose $z \in A_i \cap A_j$, and take an edge $(z,w)$ in $G$.
%Now $z\notin V_t$ and thus, since $T$ is a vine, we have that $w\in V_t$.
%Thus, any vertex in $Z$ forms a singleton component in $G-V_t$.
%
%Next observe that there are no edges between vertices in $A_i$ and $A_j$.
%To see this, consider $u\in A_i$ and  $v\in A_j$. Hence $u\notin V_t$ and 
%$v\notin V_t$. Now suppose that $(u,v)$ is an edge of $G$. Then $T_u\cup T_v$ is a
%a subtree of $T$ and $(u\cup v)\cap W_x \neq \emptyset$, a contradiction.
%
%These observations imply that the graph corresponding to our vine is of the form
%shown in Figure \ref{dec}. 

% 
% \begin{figure}[h] 
% \unitlength1cm
% \begin{center}
% \includegraphics[height=8cm, width=5cm]{decomp.pdf}
% \caption{\label{dec} A decomposition.}
% \end{center}
% \end{figure}

 \begin{lemma}
Let $(T,\l)$ be a vine decomposition of $G$, and let $t$ be an internal node in $T$.
Then $V_t$ is a separator in $G$.
\end{lemma}
\begin{proof}
Take any node $t\in T$ with degree $r\ge 2$. Let $T_1,T_2,\dots,T_r$ be the subtrees
formed by the removal of $t$. Let $C_i= \bigcup_{x \in T_i} V_x\setminus V_t$.
We claim that there is no edge between $C_i$ and $C_j$ for $i \neq j$ in $G\setminus V_t$.
To see this, take any pair of vertices $x\in C_i$ and $y\in C_j$ where $i\neq j$.
 Since neither $x$ nor $y$ are in $V_t$ their corresponding subtrees $T_x$ and
 $T_y$ can neither intersect nor be adjacent in $T$, because $T_x \subseteq T_i$ and
 $T_y \subseteq T_j$. Thus $(x,y)$ is not an edge of $G$.
 \end{proof}

% 
% The above theorem assumes integral payoffs. {\bf [What if we allow fractional payoffs as assumed
% by the core?]}

\subsection{The Thicket-Vinewidth Duality Theorem}\label{sec:thicket-duality}
Here we present the thicket-vinewidth duality theorem.
\begin{theorem}\label{thm:minmax}
The  thicket number $\tau(G)$ is equal to the vinewidth $\nu(G)$.
\end{theorem}
To prove this we apply the approach used 
by \cite{ST93} to bound the bramble number.
In particular we prove the following stronger result, which characterizes 
when a thicket can be extended to create a thicket with hitting size $k$.

\begin{lemma}\label{lem:stronger}
For any thicket $\cal{H}$ in $G$, exactly one of the following holds:
\begin{itemize}
\item[(a)]
There is a thicket $\cal{H}'$ with hitting size $k$ such that $\cal{H}\subseteq \cal{H}'$.
\item[(b)]
There is a vine decomposition $(T,\l)$ of $G$ such that for
any node $s\in T$ with $|V_s|\ge k$: \\
(i) $s$ is a leaf in $T$, and (ii) $V_s$ is not a hitting set for $\cal{H}$.
\end{itemize}
\end{lemma}

Before proving Lemma \ref{lem:stronger}, let's see why it does gives Theorem~\ref{thm:minmax}.

\begin{proof}[of Theorem~\ref{thm:minmax}]
First we show that $\nu(G)\ge \t(G)$. Take a thicket $\cal{H}=\{H_1,H_2,\dots,H_p\}$ and
a vine decomposition $T=(N,L)$ of $G$. We will show the subtree $H[T_i]$ of $T$ corresponding to each element of $\mathcal{H}$ pairwise intersect, apply the Helly Property to find a common intersection for all $H[T_i]$ and see that the label of this node at the intersection is a hitting set for $\mathcal{H}$.
%First assume that some set of $\cal{H}$ consists of a single vertex, say $H_1=\{v\}$.
%Because all sets in the thicket are pairwise intersecting, we then also have $v\in H_i$ for each $2\le i \le p$.
%Consequently, $\{v\}$ is a hitting set for $\cal{H}$ and, clearly, $\nu(G) \ge 1 = \t(G)$.

More precisely,
as $H_i$ is a connected subgraph of $G$, there is a tree $R_i$ in $G$ spanning the 
vertices of $H_i$. Now consider the subgraph of $T$ induced by $H_i$; that is 
$T[H_i]=\bigcup_{w \in H_i} T_{w}$. We claim that $T[H_i]$ is
a tree (i.e., $T[H_i]$ is connected). Take any edge $(v_r, v_s)\in R_i$. As $T$
is vine decomposition, we have that $T_{v_r}\cup T_{v_s}$ induces a connected
subtree in $T$. Extending this argument over every edge of $R_i$ implies
$T[H_i]$ is connected.

Now take any pair $H_i$ and $H_j$ in $\mathcal{H}$, $1\le i< j\le p$.
Since $\cal{H}$ is a thicket they both contain some vertex $u$ in $G$. 
It follows that $T_u\subseteq T[H_i] \cap T[H_j]$.
In particular, the $T[H_i]$, $1\le i\le p$, are pairwise intersecting subtrees of $T$.
Therefore, by the Helly Property of trees,
there is a node $t\in \bigcap_{i=1}^p T[H_i]$.  We claim that the
vertex set $V_t$ in $G$ is a hitting set for $\cal{H}$.
To see this, take any $H_i\in \mathcal{H}$. We have $t\in T[H_i]$ and, thus, there is some
$T_w\subseteq T[H_i]$ where $w\in V_t$. Since $w\in H_i\cap V_t$, the claim follows.
Since $|V_t| \le \tau(G)$, the thicket number is at most the vinewidth: $\tau(G) \le \nu(G)$. 

Next we must prove that $\nu(G)\le \tau(G)$. This follows from Lemma \ref{lem:stronger}.
To see this, let $\cal{H}=\emptyset$ be the empty thicket and $k = \nu(G)$.
Then either (a) there is a thicket $\cal{H}'$ with hitting size $k$
or (b) there is a vine decomposition $(T,\l)$ such that if  $|V_s|\ge k$ then
$s$ is a leaf in $T$ {\em and} $V_s$ is not a hitting set for $\cal{H}$. If (a) holds, 
then $\cal{H}$ is a thicket of hitting size $\nu(G)$ and so by definition of $\tau(G)$, $\tau(G) \ge \nu(G)$. We now show (b) cannot hold.
For all $s$, $V_s$ is a hitting set since $\cal{H}$ is empty.
Thus we have this restatement of (b): there is a vine decomposition $(T,\l)$ such 
that $|V_s|< k$ for all $s\in T$, $i.e.$ the vinewidth is at most $k-1 = \nu(G)-1$, a contradiction to the definition of $\nu(G)$.
\end{proof}

So now we must prove Lemma \ref{lem:stronger}. Our proof is based upon
an interpretation by Reed in~\cite{Reed02} of the bramble-treewidth duality theorem. Before
proving Lemma~\ref{lem:stronger}, let us state two lemmas.

\begin{lemma}\label{lem:twohits}
Given a thicket $\cal{H}$ with hitting size $h$.
Let $X_1$ and $X_2$ be hitting sets for $\cal{H}$. Then any separator for $X_1$ and $X_2$ contains
at least $h$ vertices.
\end{lemma}
\begin{proof}
Take a set $H_i \in \cal{H}$. The set $H_i$ is connected and intersects 
$X_1$ and $X_2$ at respectively $v_1$ and $v_2$. Thus there is a path $P_i$ in $H_i$ 
between the $v_1$ and $v_2$.
A set $S$ that separates $X_1$ and $X_2$ must disconnect
this path. This applies for every set  $H_i\in \cal{H}$, so $S$
must hit every set in $\cal{H}$ and is, thus, a hitting set. Therefore $|S|\ge h$.
We remark that, in fact, there are $h$ vertex disjoint paths 
from $X_1$ to $X_2$ in $G$. 
\end{proof}

\begin{lemma}\label{lem:separator}
Let $(\hat{T},\hat{\l})$ be a vine decomposition of $G$. 
Let  $X$ be a separator of $G$ with a component $C$ of $G\setminus X$.
Suppose there is a node $t\in \hat{T}$ such that $\hat{V}_t\cap C=\emptyset$
and every separator for $X$ and $\hat{V}_t$ contains at least $|X|$ vertices.
Then there is a vine decomposition $(T=\hat{T},\l)$ of $G[X\cup C]$ with\\
(1) $V_t=X.$ \ \ \ \ \ 
(2) $V_s\subseteq \hat{V}_s$, for each leaf $s\neq t$ in $T$.\ \ \ \ \ 
(3) $|V_s|\le |\hat{V}_s|$, for all $s\in T$.
\end{lemma}

Lemma 2.8 allows us to restrict any tree decomposition on $G$ to one on a subset of the vertices of $G$ that contains 
a special label. This subset is the union of a cutset $X$ and a component $C$ of $G-X$. The special label is $X$.

\begin{proof}
Let $A$ and $B$ be two subsets of vertices. Menger's theorem ensures 
that the maximum number of internally vertex disjoint path from $A$ to $B$ 
is equal to the minimum size of a separator for $A$ and $B$.
Thus, there are $|X|=k$ vertex disjoint paths from $\hat{V}_r$ to $X$. 
Let these
paths be $\{P_1,P_2,\dots,P_k\}$ where the endpoint of $P_i$ in $X$ is $x_i$.
Because $\hat{V}_t$ is disjoint from $C$, it follows that 
all the $P_i$ are also disjoint from $C$. 

Now consider the vine tree $\hat{T}$. Let $Q_i$ be the path in $\hat{T}$ from 
from $\hat{T}_{x_i}$ to $t$ (not including vertices of $\hat{T}_{x_i}$). 
%Clearly this path is a subset of the subtree $S_i=\bigcup \{T_v : v\in P_i\}$
%(and $Q_i\cup T_{x_i}$ is also a subtree of $T$).
We claim the desired vine decomposition $(T=\hat{T},\l)$ on $G[X\cup C]$ is
given by taking 
\begin{eqnarray*}
V_s \hspace{5pt} = \hspace{5pt} (\hat{V}_s\cap (X\cup C)) \cup \{x_i: s\in Q_i\}
\hspace{5pt} = \hspace{5pt} (\hat{V}_s\cap (X\cup C)) \cup \{x_i: P_i\cap \hat{V}_s\neq \emptyset\}
\end{eqnarray*}
Note that $s\in Q_i$ if and only if $P_i$ intersects $\hat{V}_s$ because the node $s$
separates $\hat{V}_{x_i}$ from $t$ in $\hat{T}$.
Now, first let's verify that this is a vine decomposition on $G[X\cup C]$.
Take any vertex $v\in C$. Then $T_v=\hat{T}_v$, a subtree of $T=\hat{T}$.
On the other hand, take any vertex $x_i\in X$. Then $T_v= \hat{T}_v \cup Q_i$ which, again, 
is a subtree of $T=\hat{T}$. Furthermore, for any edge $(u,v)$ in $G[X\cup C]$
we know that $\hat{T}_u$ and $\hat{T}_v$ either intersect or are adjacent in $\hat{T}$.
Therefore, $T_u$ and $T_v$ either intersect or are adjacent in $T=\hat{T}$.
 Thus we have a vine decomposition.
 
 Now let's show that the three required properties hold. 
 (1) Each path $P_i$ ends in $\hat{V}_t$. Thus, we have $x_i\in V_t$.
  Since $\hat{V}_t$ is disjoint from $C$, we obtain $V_t=X$. 
  (2) Take a leaf $s\neq t$ of $T$. The  paths $Q_i$ do not contain leaves and 
so $V_s\subseteq \hat{V}_s$. 
(3) Take a non-leaf $s$ of $T$. 
If $P_i$ intersects $\hat{V}_s$ then we have $x_i\in V_s$.
Suppose $x_i\notin \hat{V}_s$. But then there is a $y_i\in P_i\cap \hat{V}_s$ where $y_i\neq x_i$.
By disjointness, $y_i$ is not in $P_j$ for any $j\neq i$.
Moreover $y_i$ is not in $C$ since $P_i$ is disjoint from $C$. Hence, $y_i\in \hat{V}_s\setminus V_s$.
It follows that $|V_s|\le |\hat{V}_s|$. 
\end{proof}

We now have all the tools we need to prove Lemma~\ref{lem:stronger}.

\begin{proof}[of Lemma~\ref{lem:stronger}]
Assume (b) holds.
So suppose there is a vine decomposition where any node $s\in T$ with $|V_s|\ge k$ is a leaf {\bf and} does not 
provide a hitting set for $\cal{H}$. For every $\cal{H} \subseteq \cal{H}'$, the set $V_s$ is a not a 
hitting set for $\cal{H}$ either. But, as we saw when proving that $\nu(G)\ge \t(G)$ in the 
proof of Theorem \ref{thm:minmax}, the vine decomposition must contain a node $t$ that provides a 
hitting set $V_{t}$ for $\cal{H}'$. But, by $(ii)$, such a node 
has $|V_{t}|<k$. Thus, $\cal{H}$ is a thicket with hitting size at most $k-1$,
and then (a) does not hold.

We must now show that at least one of (a) or (b) holds.
We prove this by contradiction. 
In particular, take a counter-example $\cal{H}$ with the 
fewest number of hitting sets of size at most $k-1$.

Since (a) does not hold, every thicket $\cal{H}'$ containing $\cal{H}$ has hitting size at most $k-1$. 
In particular, $\cal{H}$ itself has at least one hitting set $X$ of cardinality at most $k-1$.
By assumption, no vine decomposition exists with 
property (b). Hence, $X\neq V(G)$ otherwise the trivial vine decomposition satisfies (b).
Now let $\{C_1,C_2,\dots,C_r\}$ be the connected components of $G\setminus X$. We will find vine
decompositions $(T^i,\l^i)$ of $G[X\cup C_i]$, for $1\le i\le r$,
that satisfy
\begin{enumerate}
\item[(1)] There is a node $t_i$ of $T^i$ with $V_{t_i}^i=X$.
\item[(2)]  Any node $s\in T^i$ with $|V^i_s|\ge k$ is a leaf
and $V_s$ is not a hitting set for $\cal{H}$.
\end{enumerate}
Since the $V_{t_i}^i$ are identical we may combine
the $T^i$-s by merging together all $t_i$ into a single node whose label is the union of labels of nodes we merged.
Observe that this will create a vine decomposition $T$ for $G$ satisfying (i) and (ii).
The theorem then follows. So we need to show that (1) and (2) hold for $G[X\cup C_i]$. 

First suppose $C_i$ does not intersect 
some thicket element $H_i$ in $\cal{H}$. Then we can define $T^i$ to be a tree with two nodes
$s$ and $t_i$, where $V^i_{t_i}=X$ and $V^i_s=C_i$. Trivially this is a valid vine decomposition.
Clearly, $|V^i_{t_i}|=|X|<k$ thus (1) holds. Now
$V^i_s$ is not a hitting set for $\cal{H}$, thus (2) holds.

Therefore, we may assume that $C_i$ is a hitting set for $\cal{H}$.
By the connectedness of $C_i$, we have that $\hat{\cal{H}}=\cal{H}\cup \{C_i\}$ is
also a thicket. Since (a) does not hold for $\cal{H}$, we know that
$\hat{\cal{H}}$ has a hitting set of size at most $k-1$. 
Moreover, $X$ is not a hitting set for $\hat{\cal{H}}$ since $X$ is disjoint from $C_i$.
Thus $\hat{\cal{H}}$ has fewer hitting sets of size at most $k-1$ than $\cal{H}$.
Consequently, $\hat{\cal{H}}$ is not a counterexample.
But (a) cannot hold for $\hat{\cal{H}}$ otherwise it holds for $\cal{H}\subset \hat{\cal{H}}$,
a contradiction. Hence (b) holds for $\hat{\cal{H}}$.

So take a vine decomposition $(\hat{T}, \hat{\ell})$ of $G$ satisfying (b)
for $\hat{\cal{H}}$. There must be a leaf
$t\in \hat{T}$ with $|\hat{V}_t|\ge k$ such  that $\hat{V}_t$ is a hitting set for $\cal{H}$,  
and $\hat{V}_t\cap C_i=\emptyset$. If not, $(\hat{T}, \hat{\ell})$ is a 
vine decomposition of $G$ satisfying (b) with respect to $\cal{H}$,
contradicting the definition of $G$. We will transform $\hat{T}$ into a
vine decomposition for $G[X\cup C_i]$ satisfying both (1) and (2).
Recall $X$ is a minimum hitting set for $\cal{H}$. Furthermore, $\hat{V}_t$ is
a hitting set for $\cal{H}$. 
Thus, by Lemma~\ref{lem:twohits}, every separator
for $X$ and $\hat{V}_t$ contains at least $|X|$ vertices.
Therefore, we may apply Lemma \ref{lem:separator} with $C=C_i, t=t$ to give a vine decomposition of
$G[X\cup C_i]$. 
This vine decomposition satisfies (1) and (2). To see (1), note that $t_i$ is a leaf with $V^i_{t_i}=X$.
Now $|V^i_s| \le |\hat{V}^i_s|$ and $|\hat{V}^i_s|\ge k$ only if $s$ is a leaf.
Thus, $s$ must be a leaf if $|V^i_s|\ge k$ and, hence, $V^i_s \subseteq \hat{V}^i_s$.
But $\hat{V}^i_s$ is not a hitting set for $\hat{\cal{H}}$. Consequently, either
$\hat{V}^i_s$ is not a hitting set for $\cal{H}$ or $\hat{V}^i_s\cap C_i=\emptyset$.
In the latter case, we then have that $\hat{V}^i_s \subseteq X$ and so $|\hat{V}^i_s|\le |X| \le  k-1$.
Thus (2) holds.
\end{proof}

\section{The Packing-Covering Ratio}\label{sec:ratio}
In this section, we prove that the packing-covering ratio is given exactly by 
the thicket number of the interaction graph (Theorem \ref{thm:1}).
% 
% \restatethm{\ref{thm:1}}{
% For any interaction graph $G$, the packing-covering ratio satisfies
% \[ \tau(G) \ \ \leq_\exists\ \  \frac{\kappa(\mathcal{G})}{\rho(\mathcal{G})} \ \ \leq_\forall \ \ \tau(G)\]
% }
To prove this, let's consider the lower and upper bounds separately.
\begin{theorem}
For any interaction graph $G$, the packing-covering ratio satisfies
\[ \tau(G) \ \ \leq_\exists\ \  \frac{\kappa(\mathcal{G})}{\rho(\mathcal{G})} \]
\end{theorem}
\begin{proof}
Given $G$, let $\mathcal{H}=\{H_1, H_2, \dots, H_p\}$ be a thicket with maximum 
hitting size $\t(G)$. We define a game $\mathcal{G}$ using the following valuation function:
$$
v(S) = 
\begin{cases}
1 & \quad \mbox{if $S\in \cal{H}$}\\
0 & \quad \mbox{otherwise}
\end{cases}
$$
Recall each set $H_i \in \mathcal{H}$ is connected. Thus $v$ is a valid valuation function
for a coalition game over the interaction graph $G$.
Furthermore, the sets $\{H_1, H_2, \dots, H_p\}$ are pairwise-intersecting.
Thus any partition $\mathcal{S}$ of the agents can include at most
one set from $\cal{H}$. So $\rho(\mathcal{G})= 1$.
On the other hand, with integral payoffs, we must provide
a dollar to at least one agent in each coalition in $\cal{H}$.
The cheapest way to do this is to give a dollar to each agent in
a minimum hitting set for $\cal{H}$. Thus
$\kappa(\mathcal{G})=\tau(\cal{H})= \t(G)$ and the packing-covering ratio
is exactly $\t(G)$.
\end{proof}

\begin{theorem}\label{thm:upper}
For any interaction graph $G$, the packing-covering ratio satisfies
\[ \frac{\kappa(\mathcal{G})}{\rho(\mathcal{G})} \ \ \leq_\forall\ \  \tau(G) \] 
\end{theorem}
\begin{proof}
Take any game $\mathcal{G}$ with valuation function $v$ over an interaction graph $G$.
Let $T=(N,L)$ be a vine decomposition of $G$. We may assume each label in the vine decomposition 
has size $\tau(G)$; if not, simply add vertices from the labels of adjacent nodes. Root the (vine) tree $T$ at an 
arbitrary node $r$. In turn, we may now consider each subtree of $T$ to be rooted at its (unique) 
node closest to the root $r$.
We claim that, for each coalition $Q$, the nodes $T(Q)=\bigcup_{v\in Q} T_v$ induce a connected graph
in $T$. Indeed, by viability of the coalition, we know $G[Q]$ is connected. Thus, for every edge $(u,v)$ in the
subtree $G[Q]$, we have that $T_u \cup T_v$ is connected, by definition of a vine decomposition. 
Then, since connectivity is transitive, $T(Q)$ induces a subtree in $T$.

Hence, we may define the \emph{root of a coalition}, $t_Q$, to be the root of $T(Q)$.
We are ready now to describe a payment allocation ${\bf x}$ that proves the theorem. To simplify the 
analysis, instead of allocating values to agents directly, we have an allocation $x_{i,t}$ for each agent $i$ 
and each node $t$ in $T_v$. The total allocation for agent $i$ is then simply $x_i = \sum_{t\in T_i} x_{i,t}$.
%starting with ${\bf x} = {\bf 0}$
We work bottom-up from the leaves to the root, and allocate to all vertices in a label $\ell(t)$ in turn.
At a node $t$, for each coalition $Q$ whose root is $t$, we compute the total amount $x(Q, t)$ 
allocated to $Q$ in descendants of $t$ and the residual value $r(Q, t) = \max(v(Q) - x(Q, t), 0)$ we 
still need to add to $Q$ to create a valid allocation. If $Q^*_t$ is the coalition of maximum residual 
value, we set $x_{v,t} = r(Q^*_t, t)$ for every agent $i$ in $\ell(t)$.

By the choice of $Q^*_t$, 
we can conclude (simply by looking only at allocations for $t$ and its descendants) 
that $x(Q)\ge v(Q)$, for every coalition $Q$ with root $t$. 
Furthermore, if $r(Q^*_t)$ is positive then $x(Q^*) = v(Q^*)$. 
%Since we only add to our allocation ${\bf x}$ afterwards, this remains true until the end of our allocation. 
Thus, since every coalition has a root, $x(Q) \ge v(Q)$ for every feasible coalition $Q$. 
Now let's bound the total cost of the allocation.
By construction, the cost of our allocation is $\sum_t r(Q^*_t, t)\cdot |\ell(t)|$.
Therefore, since all labels have size $ |\ell(t)|=\tau(G)$, we have that
$$\kappa(\mathcal{G)} \le \tau(G)\cdot \sum_t r(Q^*_t, t)$$ 
Therefore, it suffices to prove that
$$\sum_t r(Q^*_t, t)\le \rho(\mathcal{G)}$$
To do this we construct an integral packing $\mathcal{Q}$ of coalitions as follows. Consider the nodes of $T$ from root to leaves 
(in a postorder traversal). Initially no node is marked as deleted. At a node $t$ of $T$, if $t$ is not 
marked as deleted and the residual $r(Q^*_t, t)$ is positive, add $Q^*_t$ to $\mathcal{Q}$ and mark 
all nodes in $T(Q^*_t)$ as deleted. Otherwise, mark $t$ as deleted.
We bound the packing value of $\mathcal{Q}$ inductively using a potential function. This potential 
is defined as the total allocation in remaining nodes ($i.e.$, nodes that are not marked as deleted). 
Initially this potential is $\tau(G)\cdot \sum_t r(Q^*_t, t)$. 
When we add an element to $\mathcal{Q}$ the potential drops by $\tau(G)\cdot x(Q^*_t) = \tau(G)\cdot v(Q^*_t)$, whilst the 
value of the packing $\mathcal{Q}$ increases by $v(Q^*_t)$. By the end, every node is marked as deleted.
Consequently, our potential functions is zero and, so, the value of our packing has increased by exactly 
$\frac{1}{\tau(G)}\cdot \tau(G)\cdot \sum_t r(Q^*_t, t)  = \sum_t r(Q^*_t, t)$. Thus $\mathcal{Q}$ has the
desired value. It remains to verify the coalitions in $\mathcal{Q}$ are disjoint. So take any two coalitions,
say $Q^*_{t_1}$ and $Q^*_{t_2}$. 
If $t_1$ and $t_2$ do not form an ancestor-descendant pair in $T$ then $Q^*_{t_1}$ and $Q^*_{t_2}$
are disjoint since the nodes containing vertices of $Q^*_{t_1}$ are in the tree rooted in $t_1$
while those of $Q^*_{t_2}$ are in the tree rooted in $t_2$.
So, without loss of generality, assume that $t_2$ is an ancestor of $t_1$. It follows that
$t_1\notin T(Q^*_{t_2})$, otherwise $Q^*_{t_1}$ would not have been selected. Thus 
$Q^*_{t_1}$ and $Q^*_{t_2}$ are again disjoint.
\end{proof}

We conclude this section with a brief discussion on computational implications.
Computing the treewidth of a graph $G$ is an NP-hard problem~\cite{ACP87}. 
Whilst no formal proof is given here, the NP-completeness arguments 
can be extended to vinewidth. The existence of a constant-factor approximation algorithm 
for the treewidth running in polynomial time remains an important open problem. 
There does, however, exist a constant-factor approximation algorithm for the treewidth in FPT time
parameterized by the treewidth \cite{Bodl88}. 
Since, by Theorem \ref{thm:approx}, the vinewidth is within a factor two of the treewidth, this provides a 
constant-factor approximation algorithm for the vinewidth of the graph.

\section{VC-Dimension}\label{sec:VC}
A ubiquitous measure of the complexity of a set-family is its VC-dimension \cite{VC71}. 
As we have a set-family derived from the interaction graph $G$, it is natural to ask whether 
we can relate the packing-covering ratio to the VC-dimension. We explore this question in this section.

First, recall the definition of VC-dimension.
Given a ground set $I$ and a collection $\mathcal{R}=\{R_1,\dots, R_m\}$
of subsets of $I$,
we say that $X\subseteq I$ is {\em shattered} by $\mathcal{R}$
if, for all $Y\subseteq X$, there exists some $R_j\in \mathcal{S}$ such that $Y=X\cap R_j$.
The {\em VC-dimension} of $(I, \mathcal{R})$ is then the maximum cardinality of a
shattered set.

%Given a set of agents $I$ and a set of coalitions $\mathcal{S}=\{S_1,\dots, S_m\}$, 
%we say that $X\subseteq I$ is {\em shattered} by $\mathcal{S}$
%if, for all $Y\subseteq X$, there exists some $S_j\in \mathcal{S}$ such that $Y=X\cap S_j$.
%The {\em VC-dimension of the coalition game} $(I, \mathcal{S})$ is the maximum cardinality of a
%shattered set.
Interestingly, for any set-family, in a simple game\footnote{A \emph{simple game} is a coalition game where the value of
every coalition is either $0$ or $1$.} the primal integrality gap can also be upper bounded by the
VC-dimension. Specifically, %\citet{HW86} proved the following:
\begin{theorem}Haussler, Welzl~\cite{HW86}\label{thm:VC}
Let  $\mathcal{R}$ be a set-family with VC-dimension~$d$. Then the primal integrality gap of any simple game $\mathcal{G}$ whose 
viable coalitions are $\mathcal{R}$ satisfies 
$$\frac{\kappa(\mathcal{G})}{\kappa^f(\mathcal{G})} \ \ \le\ \ d\cdot \log d$$
\end{theorem}
We can strengthen this result when the family of coalitions is induced by an interaction graph $G=(V,E)$. 
 In particular, let the {\em graphical set family}
 $\mathcal{S}= \{S_1,\dots, S_r\}$ of $G$ be the set of {\em all} connected induced subgraphs of $G$.
 We then define the \emph{VC-dimension} of the graph $G$ to be the 
 VC-dimension of the graphical set family of $G$.
\begin{theorem}\label{thm:VC1}
Let $G$ be a graph with VC-dimension $d$. Then, restricting to simple games, the packing-covering ratio satisfies
$$\frac{\kappa(\cal{G})}{\rho(\cal{G})} \ \ \le_\forall\ \  d+1$$
\end{theorem}
\begin{proof}
%For a contradiction, assume there is a game $\cal{G}$ on $G$ with 
%with $\frac{\kappa(\cal{G})}{\rho(\cal{G})} > d$.
% 
% First assume that $\rho=1$. Thus all the coalitions of $\cal{G}$ denoted by $\cal{S}$ pairwise intersect. 
% Let $X$ be a hitting set of $\cal{S}$ of minimum size. For every $x \in X$, there exists
% $S_x \in \cal{S}$ such that $S_x \cap X = \{ x \}$. Let $\cal{S}' = \cup_{x \in X} S_x$.
% To prove that the VC-dimension of $G$ is at least $|X|$, let us show that $X$ can be shattered.
% Let $X' \subseteq X$. The set $Y= \cup_{x \in X'} S_x$ is connected and satisfies $Y \cap X = X'$.
% Thus the VC-dimension of $G$ is at least $|X|$ and then $\kappa(\cal{G}) \leq d$.
%Assume now that $\rho \geq 2$.
Let $\cal{G}$ be a coalition game over interaction graph $G$. Let $\cal{R}$ be the set of coalitions of $\cal{G}$.
Note that $\cal{R}$ is a subset for the graphical set family $\cal{S}$ of $G$. 
Let $X$ be a minimum hitting set of $\cal{R}$.
% 
% Let $X$ be a minimum hitting set for the graphical set family $\cal{S}$ of $G$. 
% Thus $X$ is a hitting set for any game $\cal{G}$ on $G$ and so $\kappa(\cal{G}) \le |X|$. 
We want to show that $\rho(\cal{G}) \ge \frac{1}{d+1}\cdot |X|$.
Now by the minimality of $X$ there is a justifying coalition $R_x\in \cal{R}$ for each vertex $x\in X$;
specifically, there exists a coalition $R_x$ such that $R_x \cap X =\{x\}$. Let $\cal{J}\subseteq \mathcal{R}$ be the set of
justifying coalitions. So $|\cal{J}|=|X|$.

First, assume there is a set $R^*\in \cal{J}$ that intersects $k$ of the other justifying coalitions. 
 Let $\cal{J}'=\{R_{x_1},\ldots,R_{x_k}\}$ be the collection of $k$
 justifying coalitions that intersect $R^*$. We now show $Y=\{x_1,\ldots,x_k\}$ is shattered for the graphical set family $\cal{S}$. We claim that for any $Y' \subseteq Y$, the set $R'=R^* \cup \bigcup_{x \in Y'} R_{x}$ intersects $Y$ in exactly $Y'$. The set $R'$ is connected 
since every $R_x$ is connected and, by construction, each of them intersects the connected set $R^*$. 
Thus $R'$ is in the graphical set family $\mathcal{S}$
and is connected. Since $R^* \cap Y$ is empty, we have $R' \cap Y = Y'$, as desired. So $Y$ can be shattered
(for the graphical set family $\cal{S}$), and thus $k\le d$
since the VC-dimension of $G$ is at most $d$.

Consequently, we may assume that every set $\cal{J}$ intersects at most $d$
of the other justifying coalitions. Then we can easily obtain a disjoint packing of $\frac{1}{d+1}\cdot |X|$ coalitions in $\cal{J}$.
Simply select any coalition $R$ in $\cal{J}$; then remove $R$ and the (at most) $d$ coalitions it
intersects from $\cal{J}$, and recurse. The theorem follows.
\end{proof}

But, we have seen that the thicket number gives the packing-covering ratio
exactly. Thus we should be able to upper bound the thicket number of any graph by a function of the 
VC-dimension. Indeed this is the case.

\begin{theorem}\label{thm:VC2}
Let $G$ be a graph of with VC-dimension $d$. 
Then the thicket number $\tau(G)$ is at most $d$.
\end{theorem}
\begin{proof}
Let $\mathcal{H}= \{H_1,\dots, H_p\}$ be a thicket with hitting number
$\tau(G)$, and let $X$ be a minimum hitting set for $\mathcal{H}$.
We claim that $X$ is shattered. To see this take any $x\in X$.
There exists some justifying set $H_x\in \mathcal{H}$ such that $H_x\cap X= \{x\}$, otherwise
$X$ is not minimal. Now take any $Y\subseteq X$.
Without loss of generality, let $Y=\{x_1, x_2, \dots , x_r\}$. 
Now each $H_{x_i}$ induces a connected graph. Moreover, the sets $H_{x_i}$  pairwise-intersect.
Thus $W=\cup_{i=1}^r H_{x_i}$
also induces a connected graph. Thus $W$ is in the graphical set family $\mathcal{S}$.
Since $H_{x_1}\cap X= \{x_i\}$ we have that 
$$X\cap W = X\cap \left(\cup_{i=1}^r H_{x_i}\right) = \cup_{i=1}^p \left(X\cap B_{x_i}\right) = \cup_{i=1}^r x_i =Y$$
Thus $X$ is shattered.
\end{proof}
Combining Theorem \ref{thm:VC2} with Theorem \ref{thm:1}, we obtain the following strengthening of
Theorem \ref{thm:VC1}, which also applies to non-simple games.
\begin{corollary}\label{cor:VC2}
Let $G$ be a graph with VC-dimension $d$. Then the packing-covering ratio satisfies
$$\frac{\kappa(\cal{G})}{\rho(\cal{G})} \ \ \le_\forall\ \  d$$
\end{corollary}
Of course, Corollary \ref{cor:VC2} must give a weaker bound than Theorem \ref{thm:1}.
Indeed, the VC-dimension of a graph $G$ can be arbitrarily
larger than its thicket number. To see this, consider a star graph with $n$ edges. Since
a star is a tree, the vinewidth of $G$ is equal to $1$ by Theorem~\ref{thm:approx}.
But the VC-dimension of the star is $n$ since any subset of leaves of the star can
be shattered using the graphical set family of $G$.

\section{The Dual Integrality Gap}\label{sec:dual-gap}

We now consider the integrality gaps of the primal and dual linear programs.
Clearly, these gaps are at most the packing-covering ratio, and thus at most the thicket number
$\tau(G)$. It is conceivable, however, that the integrality gaps could be much smaller
than the thicket number. In Section \ref{sec:primal-gap}
we will consider the primal integrality gap. 
In this section, we examine the dual integrality gap. Recall that this gap determines the
multiplicative least-core and measures the relative cost of stability.

Before quantifying the dual integrality gap in graphical coalition games, 
it is informative to give practical interpretations of the integral and fractional packing numbers.
Suppose there is a one unit interval of time, and at any point in time an agent can choose to work for
any coalition it belongs to. Thus, each agent partitions the interval into sub-intervals associated with 
assorted coalitions. Now consider two different determinants for whether a coalition is productive. 
First, suppose that a coalition can only function if its members meet together -- thus the coalition can generate wealth 
only if its members are working for it
simultaneously. In this setting, at any point in time $t$, the functioning coalitions are disjoint. It then follows,
by superadditivity, that it is best if the grand coalition meets at time $t$. Since this argument holds
for any $t$, the optimal solution is that the grand coalition to meet for the entire unit of time ($i.e.$, each agent
contributes all its time to the grand coalition).
Thus, we obtain the integral packing solution $\rho(\mathcal{G})$. Second, instead suppose
that it is not necessary for a coalition to convene simultaneously. The productivity of a coalition is then
determined by the minimum time contribution of one of its members.
So coalition members may work at different times and, in this setting, it is possible that more wealth can be 
generated if the agents fractionally allocate their time amongst multiple coalitions. Indeed the optimal solution
now is the fractional packing solution $\rho^f(\mathcal{G})$.

The main result of this section is the following:
\begin{theorem}\label{thm:dualintegralitygap}
There exist $c,\delta >0$ such that for any interaction graph $G$, the dual integrality gap satisfies:
 $$c \cdot \tau(G)^\delta \ \ \leq_\exists\ \  \frac{\rho^f(\cal{G})}{\rho(\cal{G})} \ \ \leq_\forall\ \  \tau(G)$$
\end{theorem}
The upper bound follows from Theorem \ref{thm:1}. Thus it remains for us to prove the lower bound.
To do this, we will apply some important results from graph minor theory.
A graph $H$ is a {\em minor} of a connected graph $G$ if the vertices of $G$ 
can be partitioned into $|V(H)|$ non-empty connected subgraphs such that 
if $(u,v)$ is an edge of $H$ then there exists an edge of $G$ with one endpoint in the subgraph corresponding 
to $u$ and one endpoint in the subgraph corresponding to $v$. Robertson and Seymour~\cite{RS86} proved
that every graph of treewidth $k$ admits a grid minor of size $f(k)$,
that is, the $f(k) \times f(k)$ grid is a minor of every graph of treewidth at least $k$.
Their bound $f$ was improved several times over the years,
but only recently was a polynomial bound 
obtained by Chekuri and Chuzhoy~\cite{CC13}. They prove that there exist 
$c'$ and $\delta>0$ such that every graph $G$ has a grid minor of size at least $c' \cdot \omega(G)^\delta$. 
Since $\tau(G) \leq 2\omega(G)$ by Theorem~\ref{thm:approx}, this implies that there exist $c$ 
and $\delta$ such that every graph has a grid minor of size $c \cdot \tau(G)^\delta$.
Now, before completing the proof of Theorem~\ref{thm:dualintegralitygap}, let us learn more
about the dual integrality gap of grids.
\begin{lemma}\label{lem:dualgrid}
The dual integrality gap of an $n \times n$ grid $R_n$ satisfies
 $$\frac12 \tau(R_n)\ \ \leq_\exists\ \  \frac{\rho^f(\cal{G})}{\rho(\cal{G})}$$
\end{lemma}
\begin{proof}
Take the grid $R_n$ and define a thicket $\mathcal{H}=\{H_1, H_2, \dots, H_n\}$ as
follows. The set $H_i$ is the set of vertices in the $i$th row or in the $i$th column.
Thus each $H_i$ is viable, and $H_i$ and $H_j$ intersect at two vertices
in the grid -- namely, the vertices with grid coordinates $(i,j)$ and $(j,i)$.
We create a simple game $\mathcal{G}$ over the grid $R_n$ by assigning
value one to those coalitions in $\mathcal{H}$. Clearly $\rho(\mathcal{G})=1$
as the sets in $\mathcal{H}$ are pairwise intersecting. On the other hand, $\rho(\mathcal{G})\ge \frac12 n = \frac12 \tau(R_n)$.
This is easy to see because we may assign a fractional value of $\frac12$ to each set in $\mathcal{H}$. Every grid vertex 
is in exactly two sets of $\mathcal{H}$, so this is a valid fractional packing.
\end{proof}

We are now ready to complete the proof of Theorem~\ref{thm:dualintegralitygap}.
\begin{proof}[of Theorem~\ref{thm:dualintegralitygap}]
Let $G$ be a graph of vinewidth $\nu(G)=\tau(G)$. 
We may assume that $G$ is connected. 
By~\cite{CC13}, there exist $c$ and $\delta$ such that $G$ admits a grid minor $R_k$ of size at least $k=c \cdot \tau(G)^\delta$. 
We wish to apply Lemma~\ref{lem:dualgrid} to show the existence of a game with dual integrality gap at
least $\frac12k$.
We define the thicket $\mathcal{H}=\{H_1, H_2, \dots, H_k\}$ as before, except now each node in the grid 
minor $R_k$ corresponds to a connected subgraph of the original graph $G$.
Formally, denote by $u_{i,j}$ the vertices of the grid of size $k \times k$ (where $i$ refers to the row and $j$ to the column
of the vertex $u_{i,j}$) and by $X_{i,j}$ the connected subset of $G$ which is assigned to $u_{i,j}$.
For every $i \leq k$, let $H_i= \bigcup_{j=1}^k X_{i,j} \cup \bigcup_{j=1}^k X_{j,i}$. The coalition $H_i$ 
is viable by the definition of a minor. 
Again, we take the simple game $\cal{G}$ where only the coalitions in  $\mathcal{H}$ are given value $1$.
For every $i \neq j$, the sets $H_i$ and $H_j$ intersect
since both $Y_i \cap Y_j= X_{i,j} \cup X_{j,i}$. Thus $\rho(\cal{G}) =1$ 
and $\rho^f(\cal{G}) \geq \frac{k}{2}$.
\end{proof}
Of course, any subsequent improvement in the polynomial function of~\cite{CC13} will give an
improved lower bound for~\ref{thm:dualintegralitygap}. Nonetheless, this
grid-minor method {\em cannot} provide a linear lower bound.
We prove this by considering the class of clique graphs.
\begin{lemma}\label{lem:dualclique}
The dual integrality gap of the clique $K_n$ satisfies
$$\sqrt{\frac{\tau(K_n)}{2}}-1 \ \ \leq_\exists\ \  \frac{\rho^f(\cal{G})}{\rho(\cal{G})} \ \ \leq_\forall\ \  \sqrt{8}\cdot \sqrt{\tau(K_n)}$$
\end{lemma}

To prove this lemma, we will use the same idea as Lemma~\ref{thm:dualintegralitygap} for the lower bound 
and the following lemma for the upper bound.
\begin{lemma}\label{lem:sqrtupperdig}
For any interaction graph $G$ on $n$ vertices, the dual-integrality gap satisfies
$$\frac{\rho^f(\cal{G})}{\rho(\cal{G})} \ \ \leq_\forall\ \  2\sqrt{n}$$
\end{lemma}
\begin{proof}
Take any game $\cal{G}$ on $G$. 
Since $\rho^f(\cal{G}) = \kappa^f(\cal{G})$ by strong duality, it suffices to
show that $\kappa^f(\cal{G}) \le 2 \sqrt{n}\cdot \rho(\cal{G})$.
Call a coalition {\em large} if it contains at least $\sqrt{n}$ agents
and {\em small} otherwise. 
%Let $\rho^L(\cal{G})$ be the maximum value packing of large
%coalitions and $\rho^S(\cal{G})$ the maximum value packing of small
%coalitions. 
Now greedily select a packing of small
coalitions $\{S_1,S_2, \dots S_k\}$ as follows. Let $S_1$ be the small coalition of maximum value.
Then recursively, let $S_{i+1}$ be the small coalition of maximum value that is disjoint from $\{S_1,\dots, S_i\}$.
Allocate $v(S_j)$ to every agent in $S_j$. In addition, we allocate $\frac{1}{\sqrt{n}}\cdot v^*$ to every agent, 
where $v^*=\max_{S} v(S)$.  We claim that this allocation ${\bf x}$ is a feasible solution to the primal.
To see this take any coalition $S$. If $S$ is large then $\sum_{i\in S} x_i = \frac{|S|}{\sqrt{n}}\cdot v^*\ge v(S)$,
because $|S|\ge \sqrt{n}$. Suppose $S$ is small. If $S$ was selected in the greedy packing
then $\sum_{i\in S} x_i > \sum_{i\in S} v(S_i) = |S_i|\cdot v(S) \ge v(S)$. Otherwise, let $S_j$ be the lowest index set in the 
packing that intersects $S$. By the greedy selection mechanism we then have $v(S_j)\ge v(S)$. Therefore,
$\sum_{i\in S} x_i > v(S_j) \ge v(S)$ as each agent in $S\cap S_j$ is allocated at least $v(S_j)$ (and
there is at least one such agent).
So we have $\kappa^f(\cal{G})\le \sum_{i\in I} x_i$ and, furthermore,
\begin{eqnarray*}
\sum_{i\in I} x_i
&=& n\cdot \frac{1}{\sqrt{n}}v^*+ \sum_{j=1}^k |S_j|  \cdot v(S_j) \\
&\le& \sqrt{n} \cdot v^* +\sqrt{n}\cdot \sum_{j=1}^k v(S_j) \\
&\le& \sqrt{n} \cdot \rho(\cal{G}) +\sqrt{n} \cdot \rho(\cal{G}) \\
&\leq& 2\sqrt{n} \cdot \rho(\cal{G}) 
\end{eqnarray*}
The lemma follows.
%Amongst coalitions of $\cal{G}$ of size at least $\sqrt{n}$, select a maximum packing $B$ ($B$ may be empty). 
%Now if we allocate 1 dollar to every agent in some coalition in $B$ and $1/\sqrt{n}$ dollar all remaining agents. 
%Then we have a fractional cover as any coalition not intersecting a coalition in $B$ is of size at least $\sqrt{n}$ 
%(by maximality of $B$) and received more than $\frac{\sqrt{n}} {\sqrt{n}} = 1$ dollar.
%
%Since $|B| \le \rho(\cal{G})$ and there are at most $\sqrt{n}$ agents in each of the (disjoint) coalition in $B$, this 
%allocation costs at most $|B|\sqrt{n} + \frac{n}{\sqrt{n}} \le (\rho(\cal{G}) + 1) \sqrt{n}$.
%
%But now if all coalitions (regardless of their size) intersects some element of $B$, we can allocate 0 instead 
%of $1/\sqrt{n}$ to agents not in a coalition in $B$. This allocation costs $\rho(\cal{G}) \sqrt{n}$. On the other 
%hand if some coalition $S$ does not intersect any element of $B$, $B\cup \{S\}$ is a packing and $|B| \le \rho(\cal{G}) - 1$ 
%so our allocation actually costs $\rho(\cal{G}) \sqrt{n}$.
%
%Thus, in all cases $\kappa^f(\cal{G}) \le \rho(\cal{G}) \sqrt{n}$ as we provided an explicit fractional cover and the optimal 
%cover can only cost less.
\end{proof}

\begin{proof}[of Lemma~\ref{lem:dualclique}]
First recall that the vinewidth of the clique $K_n$ is $\lceil \frac{n}{2} \rceil$. Moreover, every coalition over interaction 
graph $K_n$ induces a connected subgraph and then is viable.
For the lower bound, we label the vertices of $K_{n}$ by coordinates $i,j$ that vary between 0 and $\sqrt{n}$ 
($i.e.$, we place the vertices in a grid). Let $H_i$ consist of those vertices in the $i$th row or $i$th column. Now, 
consider the game where each $H_i$ is a coalition of value 1. Since any two element of $H_i$ intersect, $\rho = 1$. 
Since we may fractionally choose $\frac 12$ of each $H_i$, $\rho^f \ge \frac12 \sqrt{n} = \frac12 \sqrt{n} \cdot \rho$. This 
proves the lower bound.

Since $n \le 2 \tau(K_n)$, Lemma~\ref{lem:sqrtupperdig} tells us the dual integrality gap is at 
most $2\sqrt{n}\le \sqrt{8\tau(G)}$. This proves the upper bound.
\end{proof}

%The wealth of the grand coalition is defined as the maximum value
%attainable from a disjoint collection of coalitions. The wealth generated by fractional
%packing of coalitions may be far higher {\bf [intuition...]}; regardless we show that the thicket number
%is again the appropriate parameter in that it closely bounds the {\em fractional relative cost of stability}.

\section{The Primal Integrality Gap}\label{sec:primal-gap}

To conclude, we consider the primal integrality gap. This measures the maximum ratio 
in the cost between paying the agents in integral amounts and paying in fractional
amounts.
Our first result is that the thicket number does quantify the primal integrality gap to within a 
constant factor, namely Theorem~\ref{thm:2}. The upper bound follows from Theorem \ref{thm:1},
so it suffices to show the lower bound. (Due to space considerations, all the proofs for this section                     
are deferred to the appendix.)

\begin{theorem}\label{thm:twfrac}
For any interaction graph $G$, the primal integrality gap satisfies: 
\[ \frac14 \tau(G) \ \ \leq_\exists\ \  \frac{\kappa(\cal{G})}{\kappa^f(\cal{G})}\]
\end{theorem}

\begin{proof}
Take a graph $G=(V,E)$ with thicket number $\tau(G)$ and let $\mathcal{H}=\{H_1,\ldots, H_p \}$ 
be a thicket with a minimum hitting set of size $\tau(G)$. Let $X$ be a minimum cardinality hitting set for $\cal{H}$. 
Now consider the following coalition game $\cal{G}$ over the interaction graph $G$.
For a coalition $S\subseteq V$, we set $v(S)=1$ if there exists a family $\cal{H}' \subseteq \cal{H}$ 
such that $S = \cup_{H \in \cal{H}'} H$ and $|S \cap X| \geq  \lceil \frac{1}{2}\tau(G) \rceil$.
Since, the sets in the thicket are connected and pairwise intersecting we have
that $S$ is viable; thus $v$ is a valid valuation function.

The linear program has a solution with value $\kappa^f(\cal{G})\le 2$.
To see this, consider the solution where each agent of $X$ is allocated $\frac{2}{\tau(G)}$ and 
all other agents are allocated~$0$. Because each coalition with value $1$ has
at least $\lceil \frac{1}{2}\tau(G) \rceil$ members in $X$, this is a feasible fractional solution.

To prove the primal integrality gap is at least $\frac14 \tau(G)$, we will now show that
the optimal integral solution to the primal has value at least  $\lfloor \frac{1}{2}\tau(G) \rfloor+1\ge \frac{1}{2}\tau(G)$.
Suppose not. Then the set $\mathcal{S}$ of coalitions with value $1$ must have a hitting set $Y$ of
cardinality at most $\lfloor \frac{1}{2}\tau(G) \rfloor$. Now let $\hat{H}$ be the sets in the thicket
$\cal{H}$ that are disjoint from $Y$. 
Then the minimum size of a hitting set $Z$ for the family $\hat{\cal{H}}$ is at least $\lceil \frac{1}{2}\tau(G) \rceil$. 
Otherwise $Y \cup Z$ is a hitting set of $\cal{H}$ of size less than $\tau(G)$.
But then $\hat{S} = \cup_{H \in \hat{\cal{H}}} H$ satisfies $|\hat{S} \cap X| \geq \lceil \frac{1}{2}\tau(G) \rceil$.
Thus, by definition, $\hat{S}$ is a coalition in $\cal{S}$. However  $Y$ is a hitting set of $\cal{S}$ and
$Y\cap \hat{S}=\emptyset$, a contradiction.
\end{proof}

So the primal integrality gap (for the worst game) is within a factor $4$ for any pair of interaction graphs with
the same thicket number. Recall that the packing-covering ratio is the same 
(for the worst game) for every pair of interaction graphs with the same thicket number by Theorem~\ref{thm:1}.
Is this also the case for the primal integrality gap? The answer is no.
There are graphs whose primal integrality gaps differ. In particular we will show
that the integrality gap for a clique is equal to $\frac12 \tau(G)$, up to an additive constant,
whereas the class of graphs that are ``powers of a path" have integrality gaps 
that tend to $\tau(G)$. The latter inequality ensures that the upper bound of 
Theorem~\ref{thm:2} cannot be improved in general using thicket number.

\begin{lemma}\label{lem:clique}
 Let $G=K_n$ be a clique on $n$ vertices. Then the primal integrality gap of $K_n$ satisfies
 \[ \frac12 \tau(K_n) \ \ \leq_\exists\ \  \frac{\kappa(\cal{G})}{\kappa^f(\cal{G})} \ \ \leq_\forall\ \  \frac12 \tau(K_n)+1 \]
\end{lemma}

\begin{proof}
Recall that the vinewidth and thicket number of the clique $K_n$ equal $\lceil \frac{1}{2}n \rceil$.
We will prove that, for any coalition game $\cal{G}$ on the clique, the integrality gap is
at most $\frac{1}{4} n+1 \leq \frac12\tau(G)+1$.
Observe that every coalition $S$ over interaction graph $K_n$ induces a connected subgraph and is, thus, viable. 
So, in what follows, we need not verify the viability of any coalition.
Let $\cal{S}$ be the set of coalitions in the game $\cal{G}$ with value $1$.
 
 Let ${\bf x}$ be an optimal fractional solution to the primal.
 Thus, for any coalition $S\in \cal{S}$, we have $\sum_{i \in S} x_i \ge 1$.
 So $\kappa^f(\cal{G})=\sum_{i \in V} x_i$.
Now write $\kappa(\cal{G})$ as $(1-\alpha)n+1$ by choosing the appropriate $\alpha>0$
(if there does not exist such an $\alpha$, the result is proved). Thus, there is no hitting set for  $\cal{S}$
with cardinality at most $(1-\alpha)n$.
We claim that $\kappa^f(\cal{G}) \geq 1/\alpha$.
To see, this take any set $X$ of $(1-\alpha)n$ agents. Since $X$ is not a hitting set, 
there exists a coalition $S_X\in \cal{S}$ contained in its complement $\bar{X}=V\setminus X$. 
Thus $\sum_{i \in \bar{X}} x_i \geq 1$. 
This holds for {\em every} set of agents $X$ of cardinality $(1-\alpha) n$ and there are 
${n \choose (1-\alpha) n}={n \choose \alpha n}$ such sets. Each agent appears in a $\alpha$-fraction of the complements 
of these sets. Consequently, $\alpha\cdot {n \choose \alpha n}\cdot \sum_{i \in V} x_i \ge {n \choose \alpha n}$. Therefore, 
$\kappa^f(\cal{G}) = \sum_{i \in V} x_i \ge 1/\alpha$ as claimed.
As $\kappa(\cal{G})= (1-\alpha)n+1$, the primal integrality gap is at most
\[ \alpha(1-\alpha)n+\alpha \ \ \leq\ \  2\alpha(1-\alpha)\tau(K_n)+\alpha \ \ \leq\ \  2\alpha(1-\alpha)\tau(K_n)+1.\]
This is  maximized when $\alpha=\frac{1}{2}$. Thus
$$\frac{\kappa(G)}{\kappa^f(G)} \ \ \le_{\forall}\ \  \frac12 \tau(G) +1$$
%The $\cal{O}(1)$ term comes from the omitted floors and ceils and some neglected constants. 
%It achieves the proof of the upper bound. 

This upper bound is tight, we have a matching lower bound
$$\frac12 \tau(G)\ \ \le_{\exists}\ \  \frac{\kappa(G)}{\kappa^f(G)} $$
To see this, consider the game where any coalition of size $\lceil \frac{n}{2} \rceil$ has value $1$ and
any other coalition has value $0$.
Then we have $\kappa^f \leq 2$ because allocating $\frac{2}{n}$ to each agent, gives a feasible fractional
solution. On the other hand, $\kappa = \lfloor \frac{n}{2} \rfloor + 1$, otherwise some coalition will block the
allocation.
\end{proof}

There are, however, graphs for which the upper bound of $\tau(G)$ on the primal integrality gap
is obtained. Specifically, this bound it obtained for power graphs of a path.
Let $P$ be a path on $n$ vertices. The \emph{$r$-th power of $P$}, denoted by $G=P^r$, is formed by 
connecting any pair of vertices whose distance is at most $r$ in $P$.

\begin{theorem}\label{thm:gapcliquepath}
Let $G=P^r$ be the $r$-th power of a path on $n$ vertices where $n \geq 3r$. Then $\tau(G)=r$ and,
provided $n\ge k^2(r+1)$, the primal integrality gap satisfies
$$\left(1-\frac{2}{k}\right)\cdot \tau(G) \ \ \leq_\exists\ \  \frac{\kappa(\cal{G})}{\kappa^f(\cal{G})} $$
\end{theorem}

Theorem~\ref{thm:gapcliquepath} ensures that the constant $1$ in the upper bound of Theorem~\ref{thm:2} cannot be improved.
However, as observed in Lemma~\ref{lem:clique}, this upper bound cannot be reached for every graph as there are
graphs where the primal integrality gap is upper bounded by around half the vinewidth.

\begin{proof}
%If $P$ be a path on $n$ vertices then $G=P^r$ is formed by connected any pair of vertices whose distance is at most $k$ in $P$.
To show that $P^r$ has thicket number $\tau(P^r)=r$,
let us first prove that $\tau(P^r)\leq r$. We create a vine decomposition $T=(N,L)$ as follows. 
Label the vertices in order along the path as $\{1,2,\dots, n\}$.
Then we let $T$ be a path on $\lceil \frac{n}{r} \rceil$ nodes, where each node corresponds to a set of vertices of the form 
$\{qr+1, qr+2,\dots, qr+r=(q+1)r\}$, where $0\le q \le \lceil \frac{n}{r} \rceil -1$. Thus $|T_v|=1$ and so $T_v$ is trivially connected, 
for each vertex $v\in P^r$. Moreover, for each edge $(u,v)$ in $G=P^r$, either $u$ and $v$ are in the same node of $T$ or
are in adjacent nodes. So this is a vine decomposition and $\tau(P^r)=\nu(P^r)\le r$. 

Let us now show that $\tau(P^r) \geq r$. Consider the restriction of $P^r$ to the first $3r$ vertices,
and construct a thicket $\cal{H}$ as follows. Let $A=\{1,\ldots,r\}$, $B=\{r+1,\ldots,2r\}$ and $C=\{2r+1,\ldots,3r\}$.
We put a set $H$ in $\cal{H}$ if it induces a connected subgraph, it contains 
at least one vertex from each of $A, B$ and $C$, and it contains more than $r/2$ elements from at least 
two of $A, B, C$.
Note that any pair of sets in $\cal{H}$ intersect.
Assume by contradiction that $X$ is a hitting of size less than $r$.
Since $|X| < r$, none of $A,B,C$ are completely contained in $X$. Furthermore,
$X$ contains less than half of the vertices from at least
two of these three sets. Thus the complement $\overline{X}$ of $X$ contains at least one vertex in all of $A,B,C$ 
and more than half of the vertices of two of these three sets.
Moreover $\overline{X}$ is connected since $X$ does not contain $r$ consecutive vertices of $P$. So $\overline{X} \in \mathcal{H}$, 
contradicting the fact that $X$ is a hitting set.

% 
% Assume by contradiction that the hitting set $X$ of $\cal{H}$ is smaller than $r$.
% In particular, there exists an integer $i\le r$ such that the vertices $i, i+r$ and $i+2r$ are not in $X$. But $\hat{H}=\{i,r+i,2r+i\}$ induces a 
% connected subgraph of $P^r$. Thus, $\hat{H}\in \cal{H}$, a contradiction. Consequently, $\tau(P^r) \geq r$ and so $\tau(P^r)= r$.

We now show the lower bound by constructing a coalition game $\cal{G}$ on $P^r$.
For a coalition $S\subseteq V$, we set $v(S)=1$ if $S$ is a connected subgraph of $P^r$ 
of cardinality at least $\frac{n}{k}$, and set $v(S)=0$ otherwise. We have  that $\kappa^f(\cal{G}) \le k$, since
allocating $\frac{k}{n}$ to each agent gives a feasible fractional solution.
We claim that $\kappa(\cal{G}) \geq r(k-1)$. 

Let $\mathcal{S}$ be the coalitions of value $1$.
  Now any hitting set $X$ for $\mathcal{S}$ must contain at least $r$ agents from amongst $\{1, \ldots, \lceil \frac{n}{k} \rceil + r\}$;
  otherwise, there is a connected subgraph of cardinality $\lceil \frac{n}{k} \rceil$ missed by $X$.
  Similarly, $X$ must contain $r$ agents amongst $\{\lceil \frac{n}{k} \rceil + r + 1, \ldots, 2 (\lceil \frac{n}{k} \rceil + r)\}$ and, in general, 
  $r$ agents
  amongst $\{a (\lceil \frac{n}{k} \rceil + r) + 1, \ldots, (a+1) (\lceil \frac{n}{k} \rceil + r)\}$ 
  for $0\le a\le \left\lfloor \frac{n}{\lceil \frac{n}{k}\rceil + r} \right\rfloor - 1$. 
  This implies that the number of agents in $X$ is at least
$$\left\lfloor \frac{n}{\lceil \frac{n}{k}\rceil + r} \right\rfloor \cdot r \ge (k-2)r$$
Here the inequality holds if we select $k$ such that $n\ge k^2(r+1)$. 
Hence, the primal integrality gap is 
at least $(1-\frac{2}{k})\cdot \tau(G)$.
\end{proof}

\ \\
 \noindent {\bf Acknowledgements.} The third author is grateful to Reshef Meir for introducing him to this problem.

\bibliographystyle{plain}

\end{document}